\documentclass[10pt, conference, compsocconf]{IEEEtran}
\usepackage[T1]{fontenc}
\usepackage{multirow}
\usepackage{url}
\usepackage{amsmath}
\usepackage{amssymb}
\usepackage{amsthm}
\usepackage{chngcntr}
\usepackage[usenames,dvipsnames]{xcolor}
\usepackage{subfigure}
\usepackage{tikz}
\usepackage{pgflibraryshapes}
\usetikzlibrary{shapes,automata,arrows,petri,calc,positioning}
\usetikzlibrary{decorations.pathreplacing,decorations.pathmorphing}
\usepackage[font=small]{caption}
\usepackage{graphicx}
\graphicspath{{fig}}

\newtheorem{theorem}{Theorem} 

\newcommand{\ignore}[1]{}

\usepackage[algoruled,nofillcomment,lined,longend,inoutnumbered,linesnumbered,commentsnumbered]{algorithm2e}

\title{Parallel Galton Watson Process}

\author{Olivier Bodini, Camille Coti, and Julien David
\thanks{ LIPN, CNRS UMR 7030, Universit\'e Paris 13, Sorbonne Paris Cit\'e}
\thanks{
99, avenue Jean-Baptiste Cl\'ement, F-93430 Villetaneuse, FRANCE
}

}

\author{\IEEEauthorblockN{Olivier Bodini, Camille Coti, and Julien David\footnote{O. Bodini and J. David are supported by french project ANR MetaConC, 2015-2019.}}
\IEEEauthorblockA{LIPN, CNRS UMR 7030, Universit\'e Paris 13, \\
Sorbonne Paris Cit\'e\\
Villetaneuse, France\\
\url{\{firstname.lastname\}@lipn.univ-paris13.fr}}
}

\begin{document}

\maketitle

\begin{abstract}
In this paper, we study a parallel version of Galton-Watson processes for the random generation of tree-shaped structures. Random trees are
useful in many situations (testing, binary search, simulation of physics 
phenomena,...) as attests more than 49000 citations on Google scholar. Using standard analytic combinatorics, we first give a theoretical, average-case study of the random process in order to evaluate how parallelism can be
extracted from this process, and we deduce a parallel generation
algorithm. Then we present how it can be implemented in a task-based
parallel paradigm for shared memory (here, Intel Cilk). This
implementation faces several challenges, among which efficient,
thread-safe random bit generation, memory management and algorithmic
modifications for small-grain parallelism. Finally, we evaluate the
performance of our implementation and the impact of different choices
and parameters. We obtain a significant efficiency improvement for the
generation of big trees. We also conduct empirical and theoretical
studies of the average behaviour of our algorithm. 
\end{abstract}

\section{Introduction}
\label{sec:intro}

Branching processes are very simple and natural procedures that models evolution of individuals.
Such a process is extremely popular and emerges in a lot of situations; for example, in epidemiology, where individuals correspond to bacteria or in genealogy, (the initial study of Galton-Watson was the spread of surnames), but also in physics, to simulate the propagation of neutrons during a nuclear fission. 

In this paper, we focus on two standard branching processes. Firstly, on the most common Galton-Watson process which corresponds to halt with probability $\frac{1}{2}$ or generation of $2$ sons with probability $\frac{1}{2}$. 
This process appears in computer science as a good model of random rooted planar binary trees, indeed the tree obtained by writing the genealogy of the offspring conditioned to have
a fixed number $n$ of individual is known to be uniform over the set of all binary trees of size $n$.

Secondly on the process that generates uniformly rooted planar binary increasing trees. This choice is motivated by the fact that increasing trees are central data structures and arise 
in a huge number of important algorithms such as binary search algorithms, image segmentation, natural language processing, $\ldots$

It is therefore of crucial interest to be able to sample very large branching processes efficiently. For instance, nuclear simulations need very huge sampling. Now, let us observe that this process are intrinsically parallel, due to independence of each branch. But quite surprisingly, this paper is the first attempt to describe a parallel version and to the best of our knowledge, there are currently few research results on this subject~\cite{Bressan2013} (note that we do not include random numbers as combinatorial objects). For the sequential version, let us notice that the problem can be reduced to Boltzmann sampling~\cite{DFLS04,BLR15} and in the size-conditioned case has been tackled by~\cite{R85,BBJ14} for binary trees and \cite{Marchal} for increasing binary trees.
In this paper,  we give and study the first parallel algorithm that produces binary trees and increasing binary trees with a quasi-perfect distribution of load over the cores of a multi-core processor. Up to our knowledge, this paper is the first example of analyzing in distribution in the domain of parallel algorithms. 

Although Galton-Watson processes seem to be easy to parallelize, in practice a parallel implementation is really not trivial. As a matter of fact, average-case analysis results presented in this paper show that parallelism is fine-grained: the average work done by each thread is constant (or logarithmic in the increasing case) and parametrized by a threshold. This threshold characterizes a time when it is necessary that some works are given to another thread. 

This paper is organized in three main parts.
Section~\ref{sec:alg} describes an algorithm that parallelizes the Galton Watson Processes, with some small but important implementation details.
Section~\ref{sec:compan} contains a short analysis of the main parameters of this algorithm, that is the lifetime of the threads, the peak load and the total time in fully parallel model. 
This part deals with conventional analytic combinatorics. We show that, considering an ideal context in which $\sqrt{n}$ threads can run in parallel, our algorithm can sample
a $n$-node Galton Watson tree in $\Theta(\sqrt{n})$ average time complexity.
In section~\ref{sec:implem} we present implementation details that
helped us making efficient such a fine-grain parallelism.

\section{Algorithms}\label{sec:alg}

The classical and naive implementation of a Galton Watson process can be found in Algorithm~\ref{alg:naive}.
Though this method and its parallel version are not efficient,
the author thought it would improve readability to recall its
description in order to compare
this version to the improved ones.
When processing a node, the algorithm generates a random bit. If it is equal to $0$, the node is a leaf. If it is equal to $1$,
the node is internal and has two subtrees. 

\begin{algorithm}[h!]
  \KwData{a node $n$}
  \KwResult{A binary tree enrooted in $n$}
  
  {\small
      $b \leftarrow $ draw a random bit\;
      \If{$b = 1$}{
        Add nodes $n_1$ and $n_2$ as node $n$'s children\;
        $NaiveGaltonWatson(n_1)$\;
        $NaiveGaltonWatson(n_2)$\;
      }
  } 
  \caption{NaiveGaltonWatson \label{alg:naive}}
\end{algorithm}

As one can see, it is as simple to implement as it is inefficient. Double recursivity ensures to obtain a lot of movements on the call stack.
Also, since Galton-Watson processes are branching ones, a natural way to obtain a parallel algorithm is the following:
when a new node is created, its two subtrees are managed by two different threads (the main one and a new one). 
This method seems to be inefficient: as we will prove later in section~\ref{sec:compan}, the new thread immediately stops with probability $\frac{1}{2}$ 
(the subtree is a leaf), and produces a small subtree of size $4$ in average.

A proper sequential implementation requires an iterative version of this algorithm~\ref{alg:seq2}, using a stack of nodes instead of the call stack.
As we will see in benchmark section, this already drastically improves execution time. 

\begin{algorithm}[h!]         
  \KwResult{A binary tree}
  
  {\small
    $lds_1\leftarrow $ Create a linear data structure\;
    Create a tree with root $r$\;
    Push $r$ into $lds_1$\;
    \While{$lds_1$ is not empty}{
      node $n \leftarrow pop(lds_1)$\;
      $b \leftarrow $ draw a random bit\;
      \If{$b = 1$}{
        Add nodes $n_1$ and $n_2$ as node $n$'s children\;
        Push $n_1$ and $n_2$ into $lds_1$\;
        
      }
    }
  }
\caption{IterativeGaltonWatson \label{alg:seq2}}
\end{algorithm}                

The idea of our algorithm is the following: parallel computation can improve efficiency but some of its aspects might have an overhead on the computation time.
In a multi-threaded environment, waking up a sleeping thread or creating a new thread to perform part of the computation can indeed cost some time. 
Thus, one needs to make sure that the new called thread will not halt too quickly.

In order to compute in parallel efficiently, a thread should be spawned only if it has ``enough'' work to do, in a sense that we want to improve the average size of the generated subtree which is handled by a thread.
Our idea is to use a data structure to accumulate nodes to process and spawn a new thread when it reaches a sufficient size, meaning that enough nodes are to be processed by this thread.

In Algorithm~\ref{alg:pgw}, we use two linear data structures $lds_1$ and $lds_2$ to keep track of the nodes that have to be processed.
There is two advantages in doing so. First it allows us to obtain a version of the algorithm which is theoretically iterative inside a thread, which is faster than a recursive version.
Then, it allows us to solve the aforementioned problem. 
New nodes to process are pushed in $lds_1$, until its size reaches a given threshold $t$.
New nodes are then pushed in $lds_2$. When it reaches size $t$, a new thread is spawned. This thread will manage the nodes gathered in $lds_2$.

\begin{algorithm}[ht!]
  \KwData{a linear data structure $lds_1$, a threshold $t$}
  \KwResult{A binary tree}
  
  {\small
    $lds_2\leftarrow $ Create a linear data structure\;
    \While{$lds_1$ is not empty}{
      \If{$lds_2$ is empty}{
        node $n \leftarrow pop(lds_1)$\;
      }
      \Else{
        node $n \leftarrow pop(lds_2)$\;
      }
      $b \leftarrow $ draw a random bit\;
      \If{$b = 1$}{
        Add nodes $n_1$ and $n_2$ as node $n$'s children\;
        \If{$|lds_1|<t$}{
          Push $n_1$ and $n_2$ into $lds_1$\;
        } 
        \Else{
          Push $n_1$ and $n_2$ into $lds_2$\;
          \If{$|lds_2|>=t$}{
            Start $ParallelGaltonWatson(lds_2,t)$ on a new thread\;
            $lds_2\leftarrow $ Create a new linear data structure\;
          }
        }
      }
    }
  }
\caption{ParallelGaltonWatson (first called with a linear data structure containing the tree's root)\label{alg:pgw}}
\end{algorithm}

Note that the size of $lds_1$ is at most increased by $1$ at each iteration of the while loop.
Therefore, when the algorithm starts filling $lds_2$, we have $|lds_1|=t$.

\subsection{Hybrid Algorithm}
\label{algo:hybrid}

Algorithm~\ref{alg:pgw} is slower than Algorithm~\ref{alg:seq2} when small objects are generated.
Indeed, 
the parallel version requires additional data structures whose cost is not negligible if the objects are small.
Therefore, we decided to add an algorithm which is a merge of Algorithm~\ref{alg:pgw} and Algorithm~\ref{alg:seq2}: at first the tree
is generated sequentially, and once the linear data structure reaches a given size the program switches to the parallel implementation.

\section{Complexity Analysis \label{sec:compan}}

In this section, we deal with basic analytic combinatorics~\cite{Flajolet2009} in order to produce some results on the behavior of this algorithm. The first parameter that we would like to analyze is the peak total load of the processors, that is to say the maximum number of nodes in current treatment in all the threads. This analysis uses standard results on the maximum height of a Dyck path. We can easily deduce from this result that the peak load for the sampling of a tree of size $n$ is in average in $O(\sqrt{n})$. We give more details later in this section. The second one is the time of the process assuming that we have a massive parallel computation. That is to say that all thread operates in parallel from all the others. In this model, we prove in the sequel that our algorithm runs in average in $O(\sqrt{n})$. The third significant parameter is the average lifetime of the first thread. We restrict here our attention to the case where the threshold is equal to 1, 2 or 4, the complete analysis being laborious and out of the scope of this introducing paper. The first thread has the property to have in average the largest lifetime. It is a natural mean upper bound for all the other threads. Moreover, as we can see in the sequel, its mean lifetime is asymptotically constant, and consequently the mean lifetime is asymptotically the mean lifetime of almost all the threads. This ensures a good distribution of the load.

\subsection{Peak total load}
First, we want to recall the studied model. We begin with a queue containing one node, this node leave the queue and with probability 1/2 generate zero or two new nodes. 
We want to analyze the evolution of this linear data structure given the fact that we know that the process stops (the queue becomes empty) after having generated $n$ nodes. 
This type of question arises in numerous situations, from statistical physics to urn process. 
It is a well known and classically called \emph{one-dimensional Brownian excursion}. In particular, the peak load $P$ is the maximum height of this excursion. It follows a Theta distribution (see~\cite{Flajolet2009} pp328, for definitions and details):

\begin{theorem}
The peak load $P$ of our algorithm follows  after normalization by $\dfrac{1}{2\sqrt{\pi n}}$ a Theta law with expectation $\sqrt{\pi n}.$
\end{theorem}

\subsection{Time complexity in fully parallel model. }
\label{sec:anal:time}

We just analyze two cases depending on the threshold is one or two. We have two reasons for this restriction. Firstly, experimentally, these both cases are the must convenient. Secondly, up to two, the analysis is much more tricky and cannot include in this conference version. Thirdly, due to universality of the parameters, we cannot expected great changes of behavior between one and two and the next values.

So, for threshold one and two, the time complexity problem reduces to very standard question. Indeed, in case of threshold one, it is a simple observation, that every node at level $h$ (by convention, the root is at level 0) is treated after $h$ operations. Thus, the time complexity is just the height of the generated tree. This very standard problem has been tacked by De Bruijn, Knuth, Rice~\cite{BrKnRi72,FGOR93}. We come back to the same distribution that for peak load:

\begin{theorem}
The time complexity $T$ of our algorithm with threshold 1 in a fully parallel model follows  after normalization by $\dfrac{1}{2\sqrt{\pi n}}$ a Theta law with expectation $\sqrt{\pi n}.$
\end{theorem}

The next result about threshold 2 is just a remark. Indeed, every node at level $h$ is treated after $2h-1$ or $2h$ operations. This implies:

\begin{theorem}
The time complexity $T$ of our algorithm with threshold 2 in a fully parallel model follows after normalization by $\dfrac{1}{\sqrt{\pi n}}$ a Theta law with expectation $\sqrt{\pi n}.$
\end{theorem}

Notice that the algorithm with threshold 2 is theoretically slower than this with threshold. Nevertheless, our theoretical model does not take into account the fact that spawning a thread has a non negligible cost.

\subsection{Lifetime of the first thread in the case of threshold 1 and 2.}
\subsubsection{Threshold 1.}
We are going to mark the nodes which are treated by the first thread. Using standard approach by symbolic methods as presented in~\cite{Flajolet2009}, we get the following specification for the marked class of tree: 
$$\mathcal{T}^u = \mathcal{U}\mathcal{Z}+\mathcal{Z}\mathcal{U}\mathcal{T}^u\mathcal{T} \mbox{ and } \mathcal{T} = \mathcal{Z}+\mathcal{Z}\mathcal{T}^2$$

In other words, the thread processes nodes on the most left branch of the tree. The length of the left branch of a random tree is a classical problem in combinatorics. We give here a very fine analysis.

Let $T_{z,u}=\sum_{n,k}t_{n,k}z^nu^k$ be the bivariate generating function such that $t_{n,k}$ counts the number of trees of size $n$ having a first thread of size $k$. By classical dictionary from specifications to generating functions, we know that $T_{z,u}$ is given by the functional equation:
$$T_{z,u} = uz+zuT_{z,u}T(z).$$    

But $T(z)$ is nothing but the generation function of rooted binary trees, so $T(z)={\frac {1-\sqrt {1-4\,{z}^{2}}}{2z}}$. We directly deduce that $T_{z,u}={\frac {2uz}{2-u+u\sqrt {1-4\,{z}^{2}}}}.$ 

\noindent
Using guess and prove strategy, we can easily derive that when $n$ is odd $t_{n,k}= \displaystyle{{\frac {2(k-1)}{n-1}}{n-k-1\choose \frac{n-3}{2}}}$ and $t_{n,k}=0$ otherwise. This strategy proceeds as follows: firstly, calculate the first values of $t_{n,k}$, factorize them, and observe that no large divisors appear. Generally, this means that the values are product of factorials. Quite easily, we find that it should be of the shape $\displaystyle{{\frac {2(k-1)}{n-1}}{n-k-1\choose \frac{n-3}{2}}}$. Secondly, we just have to prove that our guessing is correct. For this, we extract a system of linear recurrences that allowed to build the $t_{n,k}$ from the function equation $T_{z,u} = uz+zuT_{z,u}T(z).$ Indeed, this equation is algebraic, so it is also holonomic and verifies the differential equations: 
$$\left\lbrace\begin{array}{l}
 \left( {u}^{2}{z}^{2}-8\,{z}^{2}u+8\,{z}^{2}+3\,u-3 \right) t \left( 
z,u \right) + \left( 4\,{u}^{2}{z}^{6}-{u}^{2}{z}^{4}-4\,u{z}^{4}+4\,{
z}^{4}+{z}^{2}u-{z}^{2} \right) {\frac {\partial ^{2}}{\partial {z}^{2
}}}t \left( z,u \right) + \left( 8\,{u}^{2}{z}^{5}-{u}^{2}{z}^{3}+8\,u
{z}^{3}-8\,{z}^{3}-3\,zu+3\,z \right) {\frac {\partial }{\partial z}}t
 \left( z,u \right) \\ 
t(z, u)+(2*u^2*z^2-4*u*z^2)*(diff(t(z, u), u))+(4*z^3-z)*(diff(t(z, u), z)) \\ 
\end{array}\right.   
$$
 and the coefficients follows a P-recurrence. 

We also extract a linear system of recurrences from the binomial expression: $$\left\lbrace\begin{array}{l}
\left( {k}^{2}-km-k \right) t_{2m+1,k} = \left( {k}
^{2}-2\,km-k+2\,m \right) t_{2m+1,k+1} \\ 
\left( {k}^{2}-4\,km+4\,{m}^{2}-3\,k+6\,m+2 \right) t_{2m+1
,k}=\left( {m}^{2}-km-k+3m+2 \right) t_{2m+3,k} \\ 
\end{array}\right.   
$$

  The last step is just to show that the two recurrences are equivalent.

  Now, to reach the mean lifetime of the first thread for a tree of size $n$, it suffices to observe that it corresponds to the value $M_n=\dfrac{[z^n]\dfrac{u\partial T_{z,u}}{\partial u}\vert_{u=1}}{[z^n]T(z)}$ \footnote{$[z^n]f(z)$ classically designs the coefficient of $z^n$ in the series $f(z)$}. Indeed, $[z^n]\dfrac{u\partial T_{z,u}}{\partial u}\vert_{u=1}=\sum_{n,k}k t_{n,k}z^nu^k.$ 
We have $\dfrac{\partial T_{z,u}}{\partial u}= \frac{4z}{(2-u+u\sqrt{1-4z^2})^2}$ and $\dfrac{\partial T_{z,u}}{\partial u}\vert_{u=1}={\frac {4z}{ \left( 1+\sqrt {1-4\,{z}^{2}} \right) ^{2}}}$.

To reach that $[z^n]\dfrac{\partial T_{z,u}}{\partial u}\vert_{u=1}= {2\frac {{n+1\choose n/2+1/2}}{n+3}}$, we deal with Lagrange inversion theorem~\cite{Flajolet2009} p. 732. More precisely, from the functional equation $g=z+2z^2g+z^3g^2$ followed by $\dfrac{\partial T_{z,u}}{\partial u}\vert_{u=1}$, putting $G=zg$, we get $G=z^2(1+G)^2$. So, putting $Z=z^2$, we get $G=Z(1+G)^2$, and we can directly apply Lagrange inversion theorem to yield $[Z^n]G= \frac{1}{n}{2n\choose n-1}$. The result on $[z^n]\dfrac{\partial T_{z,u}}{\partial u}\vert_{u=1}$ easily ensues.

 Moreover, $t_n=[z^n]T(z)$ is the number of binary trees that corresponds to Catalan numbers ${\frac {2{n-1\choose n/2-1/2}}{n+1}}.$ 
 
So, we first get that $M_n=\dfrac{4n}{n+3}$. 

Independently, from the exact expression of the coefficient, we also reach an exact formula for the distribution of the random variables  $L_n$ corresponding to the lifetime of the first thread in the process that return a random binary tree of size $n$. Indeed, we have $$\mathbb{P}(L_n=k)=\dfrac{t_{n,k}}{t_n}={\frac {\left( k-1 \right)  \left( n+1 \right) \sqrt {\pi }
\Gamma  \left( n-k \right) }{{2}^{n} \Gamma  \left( n/2 \right) \Gamma 
 \left( \frac{n+3}{2}-k \right) }}.
$$
Finally, using standard probabilistic approach, we obtain the limiting distribution of the $L_n$:
\begin{theorem}
Let $L_n$ be the random variable corresponding to the lifetime of the first thread in the process that return a random binary tree of size $n$ with threshold 1. Then, $\mathbb{E}(L_n)=\dfrac{4n}{n+3}$. Moreover the distribution $L_n$ converges in distribution to the random variable $X$ having $\dfrac{u^2}{(u-2)^2}$ as probability generating function.
\end{theorem}

\begin{proof}
Consider the characteristic function $\phi_n(t)=\mathbb{E}(e^{itL_n})$, using classical Flajolet-Odlysko transfer theorems~\cite{FO90}, we obtain that $$\phi_n(t)={\frac {  {{\rm e}^{2it}} }{ \left( {{\rm e}^{it}}-2 \right) ^{2}}}-12{\frac { \left( {{\rm e}^{it}}-1 \right)   {{\rm e}^{2it}}}{n \left( {{\rm e}^{it}}-2 \right) ^{4}}}+O\left( {n}^{-3/2} \right).$$
So, for every $t\in \mathbb{R}$, $\phi_n(t)$ converges pointwise to $\phi(t)= \frac {{{\rm e}^{2it}} }{ \left( {{\rm e}^{it}}-2 \right) ^{2}}$, by L\'evy's continuity theorem, this implies that $L_n$ converges in distribution to the random variable $X$ having $\dfrac{u^2}{(u-2)^2}$ as probability generating function.
\end{proof}

\subsubsection{Threshold 2.}
We are going to mark the nodes which are treated by the first thread. Using standard approach by symbolic methods, we get the following specification for the marked class of tree: $$\mathcal{T}^u = \mathcal{U}\mathcal{Z}+\mathcal{U}^3\mathcal{Z}^3+\mathcal{Z}^2\mathcal{U}^2\mathcal{T}^u_{>1}+\mathcal{Z}\mathcal{U}\mathcal{T}^u_{>1}\mathcal{Z}\mathcal{U}
+\mathcal{Z}\mathcal{U}\mathcal{T}^u_{>1}\mathcal{U}\mathcal{T}_{>1}.$$
Indeed, A marked tree in $\mathcal{T}^u$ can be recursively build as follows: if its size is 1 or 3, in this case, all the node are treated by the first thread, this corresponds to the 2 first terms in the specification. Otherwise, we have 3 possible cases, the root of the tree has 2 sons $L$ and $R$, and $|L|=1$, $|R|>1$ or $|L|>1$, $|R|=1$ or $|L|>1$, $|R|>1$ (where $|R|$ designs the size). In the both first cases, we mark the root, the singleton subtree. The remainder subtree belongs to $\mathcal{T}^u$. In the last case, we mark the root, the left subtree $L$ is in $\mathcal{T}^u$, and the right subtree $R$ is unmarked expect its root.

Now, directly from the specification, we get that $T_{z,u}=T(z,u)$ (and $T_z=T(z)$) is given by the functional equation:
$$T_{z,u}= uz+u^3z^3+2z^2u^2(T_{z,u}-uz)+zu^2(T_{z,u}-uz)(T_z-z).$$ We directly deduce that $T_{z,u}={\frac {uz \left( 2-{u}^{2}+{u}^{2}\sqrt {1-4\,{z}^{2}} \right) }{2-2
\,{z}^{2}{u}^{2}-{u}^{2}+{u}^{2}\sqrt {1-4\,{z}^{2}}}}.$ We then can easily derive that $t_{n,k}=\displaystyle{\sum_{j=0}^{\frac{k-3}{2}} {\frac {j{\frac{k-3}{2}\choose j}{n-k+j\choose \frac{n-k}{2}}}{n-k+j}}}$. 
 Again, using Odlysko-Flajolet transfer theorems, from the fact that $[z^n]\dfrac{\partial T_{z,u}}{\partial u}\vert_{u=1}=[z^n]{\frac {2z \left( 1+{z}^{2}+\sqrt {1-4\,{z}^{2}}-{z}^{2}\sqrt {1-4\,
{z}^{2}} \right) }{ \left( 1-2\,{z}^{2}+\sqrt {1-4\,{z}^{2}} \right) ^
{2}}},$ we reach that $$[z^n]\dfrac{\partial T_{z,u}}{\partial u}\vert_{u=1}= 17/2\,{\frac {\sqrt {2} \left( 1-{{\rm e}^{i\pi \,n}} \right) {2}^{n}}{
{n}^{3/2}\sqrt {\pi }}}
+O \left( {n}^{-5/2}
 \right)
.$$ In fact, with a more technical study, we can show that $M_n=\dfrac{[z^n]\dfrac{\partial T_{z,u}}{\partial u}}{[z^n]T_{z,u}}\vert_{u=1}=\dfrac{17n^2-8n+15}{n^2+8n+15}$ 

\begin{theorem}
Let $L_n$ be the random variable corresponding to the lifetime of the first thread in the process that return a random binary tree of size $n$ with threshold 2. Then, $\mathbb{E}(L_n)=\dfrac{17n^2-8n+15}{n^2+8n+15}$. Moreover the distribution $L_n$ converges in distribution to the random variable $X$ having ${\frac {{u}^{5}}{ \left( 3\,{u}^{2}-4 \right) ^{2}}}$ as probability generating function.
\end{theorem} 
\begin{proof}
Again, considering the characteristic function $\phi_n(t)=\mathbb{E}(e^{itL_n})$, and using classical Flajolet-Odlysko transfer theorems, we obtain that $$\phi_n(t)={\frac {  {{\rm e}^{5it}} }{ \left( 3\,  {
{\rm e}^{2it}} -4 \right) ^{2}}}+{\frac { 24 
 \left( {{\rm e}^{4it}} -5  {{\rm e}^{2it}}  
+4 \right)   {{\rm e}^{5it}} }{n \left( 3
  {{\rm e}^{2it}} -4 \right) ^{4}}}+O \left( {n}^{-2}
 \right).$$
So, for every $t\in \mathbb{R}$, $\phi_n(t)$ converges pointwise to $\phi(t)=\frac {  {{\rm e}^{5it}} }{ \left( 3  {
{\rm e}^{2it}}  -4 \right) ^{2}}$, by L\'evy's continuity theorem, this implies that $L_n$ converges in distribution to the random variable $X$ having $\dfrac{u^5}{(3u^2-4)^2}$ as probability generating function.
\end{proof}
%

\begin{figure}[h!]
\includegraphics[width=4cm, bb=0 0 300 300]{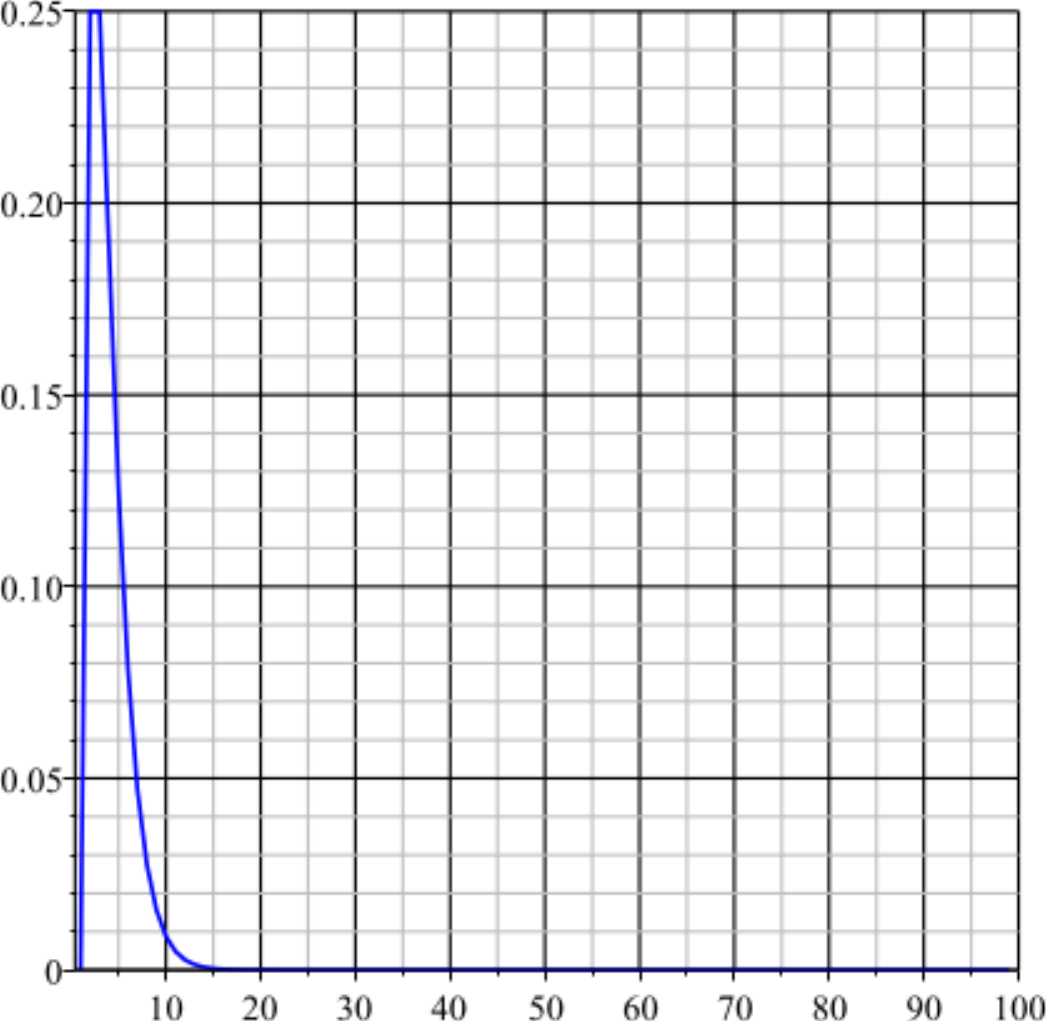}\hfill
\includegraphics[width=4cm, bb=0 0 300 300]{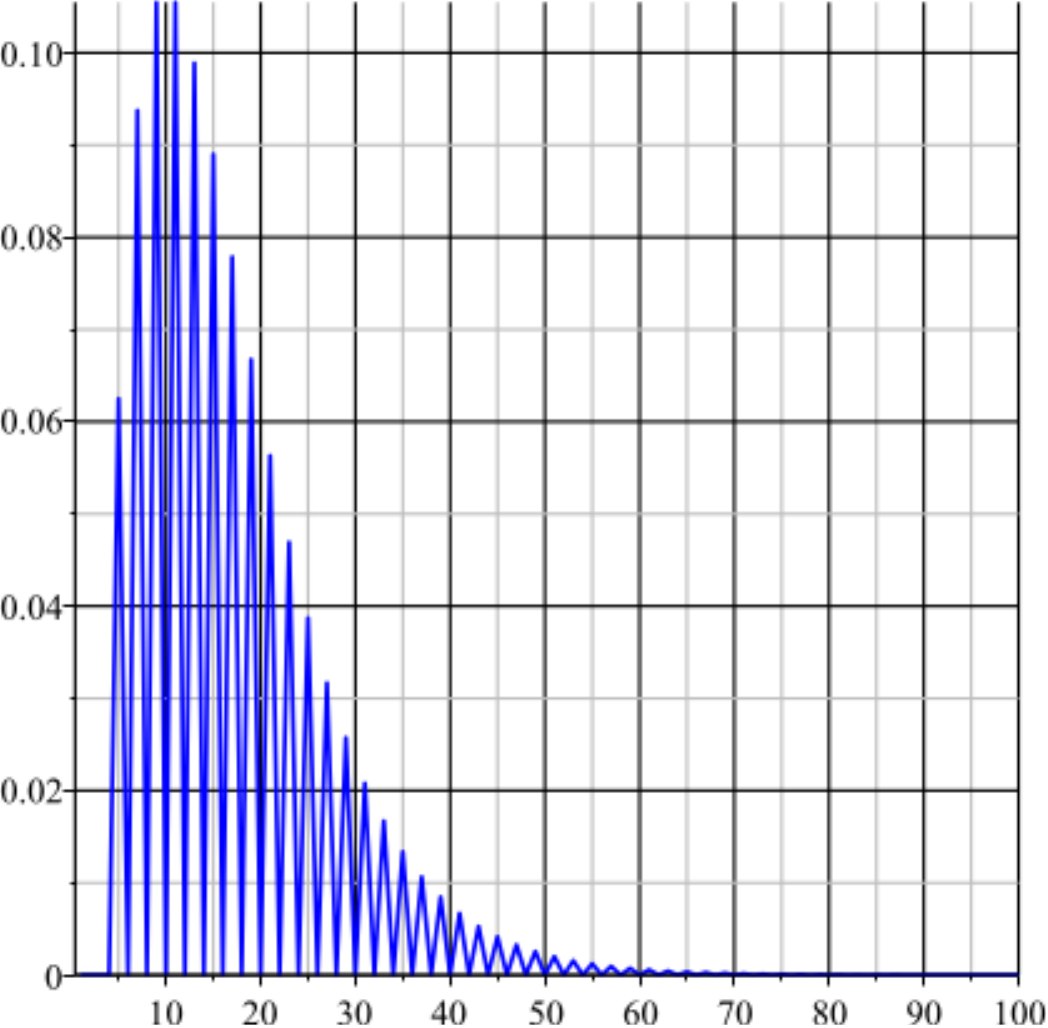}
\caption{Limiting distribution of lifetime with threshold 1 and 2}\label{fig:somefiglabel}
\end{figure}

Note to conclude this section that for every threshold the expected lifetime of the first thread is always finite (in a sense where it admits a finite limit when the size tends to the infinity). This is based on the universal shape in $(1-z/\rho)^{1/2}$ of the dominant singularity of the generating function of the mean lifetime. Nevertheless, due to the increasing complexity of the specifications, the calculation is more and more tricky and becomes humanly intractable for threshold greater than 8. In order to give an intuition of the next step, we propose without explanation the specification for threshold 4:
\begin{tiny}

$$\mathcal{T}^u={\mathcal{Z}}^{7}{\mathcal{U}}^{7}({\mathcal{T}}^{4}R_4+4 R_2+
6R_4+4R_6)
+{\mathcal{Z}}^{5}{\mathcal{U}}^{5}(4{R_2}+2{R_4})+\mathcal{X}$$
\end{tiny}
where
\begin{tiny}
$$\mathcal{X}={\mathcal{Z}}^{7}{\mathcal{U}}^{7}+2{\mathcal{Z}}^{5}{\mathcal{U}}
^{5}+{\mathcal{Z}}^{3}{\mathcal{U}}^{3}+\mathcal{Z}\mathcal{U}$$
$${R_2}={\mathcal{Z}}^{2}{\mathcal{U}}^{2} \left( 1+2{R_2}+{
R_4} \right)$$
$${R_4}={\mathcal{Z}}^{4}{\mathcal{U}}^{4} \left( {\mathcal{T}}^{4}{R_4}+4{
R_2}+6\,{R_4}+4{R_6}+1 \right)$$
\end{tiny}
and
\begin{tiny}
\begin{multline*} 
{R_6}=R_4\mathcal{T}^8\mathcal{U}^2\mathcal{Z}^6+2R_4\mathcal{T}^6\mathcal{U}^4\mathcal{Z}^6+\\
5R_4\mathcal{T}^4\mathcal{U}^6\mathcal{Z}^6+4R_4\mathcal{T}^6\mathcal{U}^2\mathcal{Z}^6+4R_4\mathcal{T}^4\mathcal{U}^4\mathcal{Z}^6+\\
6R_2\mathcal{U}^6\mathcal{Z}^6+6R_4\mathcal{T}^4\mathcal{U}^2\mathcal{Z}^6+2R_4\mathcal{T}^2\mathcal{U}^4\mathcal{Z}^6+\\
14R_4\mathcal{U}^6\mathcal{Z}^6+14R_6\mathcal{U}^6\mathcal{Z}^6+\mathcal{U}^6\mathcal{Z}^6+4R_4\mathcal{T}^2\mathcal{U}^2\mathcal{Z}^6+R_4\mathcal{U}^2\mathcal{Z}^6
\end{multline*}
\end{tiny}
By the same approach than for threhold 1 and 2, we then prove that 
\begin{theorem}
Let $L_n$ be the random variable corresponding to the lifetime of the first thread in the process that return a random binary tree of size $n$ with threshold 4. Then, the random variable $L_n$ converges in distribution to the random variable $X$ having $-{\frac {{u}^{11}}{ \left( {u}^{4}-16\,{u}^{2}+16 \right)  \left( {u}^
{6}-18\,{u}^{4}+48\,{u}^{2}-32 \right) }}$
 as probability generating function. In particular, the mean lifetime is asymptotically 69.
\end{theorem}

\section{Implementation details}\label{sec:implem}



We implemented this parallel algorithm using the task-oriented Intel
Cilk framework\cite{cilk532}. This model is
particularly well suited for recursive processes such as the Galton
Watson process. Therefore, we implemented the program in C++.


\subsection{Memory management}\label{sec:algo:mem}

One of the key issues to ensure efficiency is to avoid different threads to access the same memory zone.
Mutual exclusion is necessary to avoid race conditions and false
sharing. Race conditions happen when multiple threads are accessing
the same memory zone, with at least one of them in write
mode. They can be solved by using locks and mutexes, which are
expensive in terms of computational cost. False sharing happens when
several threads are trying to access areas of memory that distinct but
located on the same cache line. In this case, the operating system
sequentializes the memory accesses, harming the parallel
performance. False sharing can be avoided by using padding in order to
make sure that variables that are accessed by different threads are
far enough from each other.

In our implementation, most of the times nodes of the tree are only accessed by the thread that created them. 
A node contains the address of its two children. The children of a
node are necessarily created by the same thread and therefore are adjacent in memory.
We pre-allocate a memory block to each thread corresponding to future nodes. Once the block of a thread is filled, we create a new block and append
it to the previous one, in the manner of an linked list.
As a consequence, each thread has its own area of memory which is seldom
accessed by other threads. 
The only times a node can be accessed by another thread is when a linear data structure of nodes is passed from a thread $t_1$ to a thread $t_2$.
In this case, thread $t_2$ will access nodes in the linear data structures which are stored in thread $t_1$ memory space.
Though, this does not happen often and $t_1$ do not access those nodes. 
Therefore, there is few concurrency on memory accesses, and the implementation of these memory blocks makes
sure that nodes located in different memory blocks are far enough from
each other in memory not to be stored on the same cache line.
A program can generate several trees in a row reusing the memory allocated for the previous ones.

\subsection{Memory management}
As stated in section \ref{sec:algo:mem}, we implemented separate
memory blocks for each thread. Each thread allocates its own memory
blocks and accesses it. 
Moreover, to reduce the number of
 memory allocations (which are expensive system calls), we used a mass
 allocation strategy. Space for a certain number of nodes is allocated
 at once as a table and nodes from this table are used when
 necessary. 
 The global tree is represented across these blocks by pointers between
 nodes. When a node is created on a thread $T_i$ called from another
 thread $T_j$, the parent node on $T_j$ stores a pointer to the node
 on $T_i$; as a consequence, no thread performs any data access on
 another thread's memory blocks.

\subsection{Random Number Generation}

In order to obtain random bits, we had to find a pseudorandom number generator which would:
\vspace{-.3em}
\begin{enumerate}[noitemsep] 
\item
be as fast as possible,
\item
waste as little random bit as possible.
\item
have a huge period,
\item
be thread-safe,
\end{enumerate}

The naive way to draw a random bit in C/C++ would be to use the {\tt rand()\%2} instruction.
Though is technique is easy to implement, it is completely inefficient: the {\tt rand} function
is rather slow (one of its step is a huge multiplication), it is neither reentrant nor thread safe, 
its period is $2^{32}$ (which is insufficient when generating several huge objects) and finally,
this technique waste $31$ random bits since only the last one is kept.


\paragraph{Properties $3$ and $4$}
Since the classical {\tt rand} function is not adapted, we used a pseudo random number generator from the {\tt Boost} C++ libraries: the Mersenne Twister.
We used this method to generate pseudo uniform 32 bits integers.
Its period is $2^{19937}$ and using a different generator in each thread is sufficient to guarantee it is thread-safe.

\paragraph{Properties $1$ and $2$}

In order to avoid wasting random bits and accelerate the computation, we used the pseudo random integer generator to obtain a buffer of random bits.
Therefore, there is no wasted random bits, except the ones that were left unused in the buffer at the end of the program, which is negligible.
Since the random integer function is called $32$ times less, this considerably accelerates the computation.
Preserving a buffer from a call of a thread to another requires it to be stored as a global variable.
When the program starts, a tabular of buffers is generated. The number of buffers is equal to the number of threads.
Since this is a tabular which will be accessed by all the thread, we
used padding to avoid false sharing between the random number
generators. Therefore, the size of the buffers is related to the size of a cache line.

\bibliographystyle{plain}
\bibliography{gene}

\end{document}